\documentclass[compsoc,conference,a4paper,10pt,times]{IEEEtran}

\usepackage{cite}
\usepackage{amsmath,amssymb,amsfonts}
\usepackage{algorithmic}
\usepackage{graphicx}
\usepackage{textcomp}
\usepackage{xcolor}
\usepackage{booktabs}
\usepackage[hidelinks]{hyperref}
\usepackage{msc}
\usepackage{listings}
\usepackage{mathtools}

\usepackage{colortbl}
\usepackage{mathsymbols}

\newtheorem{lemma}{Lemma}
\newtheorem{theorem}{Theorem}
\newtheorem{proof}{Proof}

\def\multichoose#1#2{\ensuremath{\left(\kern-.3em\left(\genfrac{}{}{0pt}{}{#1}{#2}\right)\kern-.3em\right)}}

\definecolor{Gray}{gray}{0.3}
\definecolor{codegreen}{rgb}{0,0.6,0}
\definecolor{codegray}{rgb}{0.5,0.5,0.5}
\definecolor{codepurple}{rgb}{0.58,0,0.82}
\definecolor{backcolour}{rgb}{1,1,1}

\lstdefinestyle{mystyle}{
	frame=single,
	backgroundcolor=\color{backcolour},   
	morecomment=[f][\color{Gray}][0]{***},
	keywordstyle=\color{black}\bfseries,
	numberstyle=\tiny\color{codegray},
	stringstyle=\color{codepurple},
	basicstyle=\ttfamily\footnotesize,
	breakatwhitespace=false,         
	breaklines=true,                 
	captionpos=b,                    
	keepspaces=true,                 
	numbers=left,                    
	numbersep=5pt,                  
	showspaces=false,                
	showstringspaces=false,
	showtabs=false,                  
	tabsize=2,
}
\begin{document}

\title{Modelling Arbitrary Computations in the Symbolic Model using an Equational Theory for Bounded Binary Circuits}

\author{
	\IEEEauthorblockN{Michiel Marcus}
	\IEEEauthorblockA{\textit{TNO} \\
		Den Haag,\\The Netherlands \\
		michiel.marcus@tno.nl}
		\and
		\IEEEauthorblockN{Anne Nijsten}
	\IEEEauthorblockA{\textit{TNO} \\
		Eindhoven,\\The Netherlands \\
		anne.nijsten@tno.nl}
	\and
	\IEEEauthorblockN{Frank Westers}
	\IEEEauthorblockA{\textit{TNO} \\
		Den Haag,\\The Netherlands \\
		frank.westers@tno.nl}

}

\lstset{style=mystyle}
\lstset{escapeinside={<@}{@>}}

\maketitle

\begin{abstract}
    In this work, we propose a class of equational theories for bounded binary circuits that have the finite variant property.
These theories could serve as a building block to specify cryptographic primitive implementations and automatically discover attacks as binary circuits in the symbolic model.
We provide proofs of equivalence between this class of equational theories and Boolean logic up to circuit size $3$ and we provide the variant complexities and performance benchmarks using Maude-NPA. This is the first result in this direction and follow-up research is needed to improve the scalability of the approach.

\end{abstract}


\section{Introduction}
\label{sec:introduction}

As more and more cryptographic protocols are developed to ensure confidentiality, integrity and authenticity of our data, it is vital that these protocols themselves are secure. Numerous tools and techniques have been developed that allow us to formally verify whether these cryptographic protocols work as intended \cite{blanchet2007cryptoverif, proverif, meier2013tamarin,maude-npa}.

Formal analysis of cryptographic protocols is typically conducted under one of two common models: the computational or the symbolic model. The \emph{computational model} allows for making strong statements about the investigated system. Specifically, they can prove an upper bound on the probability than an attacker can break the system in terms of the security of the building blocks tools. However, tools in this model generally offer very little automation. 

The \emph{symbolic model} makes it possible to reason about security protocols at a higher abstraction level, which allows for a higher degree of automation of the proofs. Common vulnerabilities, such as logical flaws that compromise security properties, can be caught in the symbolic model by automated protocol verification tools \cite{SoK}.
In contrast with the computational model, the symbolic model does not reason about bit strings, but about arbitrary-length variables and models cryptographic primitives in an idealised manner. The symbolic model is therefore not sound, as an attack on the protocol abusing computational aspects of a cryptographic scheme will not be caught in the symbolic model. However, such attacks are often very complex and are also very hard to discover manually using cryptanalysis. In the symbolic model, we typically consider a \textit{Dolev-Yao adversary} \cite{dolev-yoa} that has complete control over the network. The adversary can intercept messages sent between protocol participants, record the messages and insert messages the attacker has created using their own knowledge.

 In this work, we are specifically interested in tools that do \emph{symbolic} verification. In the symbolic model, mathematical properties are modelled using an equational theory. An example is the equational theory with only the rule $dec(enc(m, k), k) = m$, where $dec$ denotes decryption, $enc$ denotes encryption, $m$ is an arbitrary message and $k$ is an arbitrary key. This equational theory models the mathematical properties of a perfectly secure symmetric encryption scheme - the only way to turn a ciphertext $enc(m, k)$ back into the message is through the rule that requires the knowledge of the key. This is obvious to see in our one-line equational system, but if the system becomes more complex, proving security statements become harder. However, if an equational system has the so-called finite variant property (FVP), which we elaborate on in section \ref{definition:fvp}, then the problem can be solved relatively efficiently. In that case there is a procedure called variant narrowing that can be used for unification during the analysis of the system \cite{variant-unification}, a technique that is commonly used by automated protocol verifiers. Hence, it is favourable for an equational system to have the FVP.
 
Examples of equational theories with the FVP are abelian groups, blind signatures, and modular exponentiation. However, several theories, such as associativity (without commutativity), commutativity (without associativity) and the homomorphic property of partially homomorphic encryption do not have the FVP \cite{homenc-fvp}. In our work, we consider equational theories for Boolean logic. The theory of Boolean logic is not FVP, because the equation that models the homomorphic property of homomorphic encryption 
\begin{align*}
\texttt{enc(x + y, k) = enc(x, k) + enc(y, k)}
\end{align*}
is similar to the distributivity rule in Boolean logic 
\begin{align*}
\texttt{(x or y) and z = (x and z) or (y and z)}
\end{align*}

and therefore has similar problems with respect to the FVP. Since there is already some literature available on equational theories with FVP for the homomorphic property of homomorphic encryption, but little on equational theories with the FVP for Boolean logic, we will refer to the former when motivating and explaining our work. Since Boolean logic encompasses more rules than just distributivity, previous literature needs to be extended in order to construct an equational theory with the FVP for Boolean logic, which is the goal of our work.

When an equational theory does not have the FVP, variant narrowing cannot be used for unification. However, it is possible to devise specific unification algorithms for non-FVP theories, which would enable automatic protocol analysis for protocols with these theories. There is a composition theory that states that variant narrowing yields a unification algorithm for the combination of any number of equational theories with the FVP \cite{variant-unification}. If we have a protocol with partially homomorphic encryption, then we have on the one hand decryption and encryption cancellation, which has the FVP, and on the other hand the homomorphic property of the encryption scheme, which does not have the FVP. Even with the dedicated unification algorithm for the homomorphic property of homomorphic encryption as introduced by Anantharaman et al. in \cite{cap_unification}, there is no known theorem that explains how it can be combined with variant narrowing to yield a unification algorithm for the combined equational theory of homomorphic encryption. Similarly, a dedicated unification algorithm for the one-sided distributivity rule of Boolean logic was proposed by Tiden and Arnborg in \cite{boolean_unification}, but it has the same issue. Additionally, Tiden and Arnborg prove that unification for the two-sided distributivity rule of Boolean logic is NP-hard. The best approach is therefore to construct an equational theory that is an approximation of Boolean logic and has the FVP.

\subsection{Contributions}
In this work, we build on main ideas in the work by Escobar et al.\cite{homenc-fvp} and create an equational theory for bounded binary circuits with the FVP. Concretely, we present the following contributions:

\begin{itemize}
	\item We introduce the first equational theory for bounded binary circuits with the FVP;
	\item We prove that this equational theory models all properties of Boolean algebra.
\end{itemize}

This equational theory serves as a building block for automatic protocol verification tooling based on narrowing to verify protocols with specific implementations. Instead of using abstract encryption schemes and digital signatures schemes, the encryption/decryption and sign/verify functionality can be modelled using binary circuits. This allows the tooling to find more complex attacks on the protocol using implementation details.

\subsection{Outline of the Paper}
In Section \ref{sec:background} we provide the necessary background information about symbolic verification and equational theories. Next, we discuss related work in Section \ref{sec:related-work}. Further, we present our strategy for modelling arbitrary binary circuits in Section \ref{sec:modelling-binary-circuit-in-maude-npa}. In Section \ref{sec:verifying-small-circuit} we formally verify a small circuit as a proof-of-concept protocol. Finally, in Section \ref{sec:results} we provide a benchmark of our approach, which is followed by a conclusion in Section \ref{sec:conclusion}.

\section{Background}
\label{sec:background}
In this section, we provide the relevant background for symbolic verification and equational theories. In this paper, we follow the notation from \cite{escobar-variant-narrowing}. The mathematical definitions in the following sections are necessary to define the finite variant property.
\subsection{Security Analysis}
A cryptographic protocol in the symbolic model is generally defined as an order-sorted rewrite theory $\RwTh = (\Sig, \EqTh, \RwRules)$, where $\Sigma$ is a signature, $\EqTh$ is an equational theory and $\RwRules$ contains the rewrite rules of the protocol. In the following sections, we elaborate further on these concepts. Intuitively, the equational theory captures the properties of the cryptographic functions, while the rewrite rules capture the communication rules in the protocol. We then perform symbolic reachability analysis to determine whether a state that constitutes a security violation can be reached given the rewrite rules in $\RwRules$ modulo $\EqTh$. If such a state is found, this constitutes an attack. This reachability analysis needs to explore an infinite state space, as the attacker can in principle send any message they want.

\subsubsection{Signature}
A signature $\Sig=((\SortSet, \leq), \FunSet)$ consists of a finite partially-ordered set of sorts $(\SortSet, \leq)$ and a finite set of function symbols $\FunSet$ over $\SortSet$. Function symbols generally represent algorithmic operations that are used within a protocol. For example, take $\FunSet=\{enc: \sort{Msg} \times \sort{Key} \to \sort{Msg}\}$ and $\SortSet=\{\sort{Msg}, \sort{Key}\}$ with $\sort{Key}\leq\sort{Msg}$.  A top sort is a sort $s \in \SortSet$ such that there is no sort $s' \in \SortSet$ for which $s < s'$. In the previous example, \sort{Msg} is a top sort.

Sorts can be though of as types of messages. Sometimes it is enough to use a single type. In this work, sorts are necessary to construct specific equational theories. They can also be used to improve the efficiency of the protocol analysis tool, because the search space for attacks becomes smaller. However, modelling all data types as a single sort gives stronger security guarantees. 

Let $\VarSet=\bigcup_{s\in \SortSet}\VarSet_s$ be the union of a mutually disjoint family of sets where each $\VarSet_s$ is a countably infinite set denoting the variables of sort $s$. The term algebra is denoted $\TermAlg{\Sig}{\VarSet}$. The elements of a term algebra are called terms. Given a term $t$, we let $\vars(t)$ denote the set of distinct variables that occur in term $t$.

A term can be just a variable $x$ of sort $s$ from $\VarSet_s$ or a more complex nested structure like $dec(enc(x, k), k)$, with $dec, enc \in \FunSet$ and $x, k \in\VarSet$. Then $\vars(dec(enc(x, k), k))=\{x, k\}$.

\subsubsection{Equations and substitutions} An equation is a tuple of terms $(t, t')$, often written as $t = t'$. Given a signature $\Sig$, a set $E$ of $\Sig$-equations induces a congruence relation on $\TermAlg{\Sig}{\VarSet}$ \cite{term-alg-extension}. We let $t =_E t'$ denote that $t$ and $t'$ are in the same congruence class according to $E$.

A substitution $\sigma$ is a mapping from a finite subset of $\VarSet$ to terms. Substitutions are extended to $\TermAlg{\Sig}{\VarSet}$ as usual. For example, let $t = dec(x, k)$ and $\sigma = \{ x \gets enc(z,k)\}$. Now $\sigma(t) = dec(enc(z, k), k)$.

There is an ordering amongst terms. For terms $t$ and $t'$, $t \leq_\EqTh t'$ if there is a substitution $\sigma$ such that $\sigma(t) =_\EqTh t'$. For example, with $t = dec(x, k)$ and $t' = y$ and $\EqTh = \{dec(enc(x, k), k) = x\}$, we have that $t \leq_\EqTh t'$, since $\sigma(t) =_\EqTh t'$ for $\sigma=\{ x \gets enc(y, k), k \gets k\}$.
 	
An equational theory $\EqTh$ is \textit{regular} if for each $t=t'$ in $\EqTh$, we have $\vars(t)=\vars(t')$. A set of equations is called \textit{sort-preserving} if for any substitution $\sigma$ and $t=t'$, $\sigma(t)$ and $\sigma(t')$ have the same sort.

There is also an ordering amongst substitutions. For substitutions $\sigma$ and $\sigma'$ and $X \subseteq \VarSet$, $\sigma \leq_\EqTh \sigma' [X]$ if for all $x \in X$, $\sigma(x) \leq_\EqTh \sigma'(x)$. Equivalently, $\sigma \leq_\EqTh \sigma' [X]$ if there exists a substitution $\sigma''$ such that for all $x \in X$, $\sigma''(\sigma(x)) =_\EqTh \sigma'(x)$. For example, let $\EqTh=\emptyset$, $\sigma = \{x \gets enc(x, k)\}$, $\sigma'= \{x \gets dec(enc(x, k), k)\}$. Then $\sigma \leq \sigma' [\{x\}]$, because $\sigma''(\sigma(x)) =_\EqTh \sigma'(x) = dec(enc(x, k), k)$ for $\sigma''=\{y \gets dec(y, k)\}$.

\subsubsection{Rewrite rules}
\label{eq_theory_background}
A rewrite rule is a pair of terms $l$ and $r$, denoted as $l \rewritesto r$, where $l\notin \VarSet$ and ${\vars(r)\subseteq\vars(l)}$. A term is in normal form if it cannot be rewritten any further. A set of rewrite rules $\RwRules$ is \textit{sort-decreasing} if for each $l\rewritesto r$  in $\RwRules$ and any substitution $\sigma$, if $\sigma(l)$ has sort \sort{s} then $\sigma(r)$ has sort $\sort{s}$. 

Given an order-sorted rewrite theory $\RwTh=(\Sig, \EqTh, \RwRules)$, we define a relation $\rewritesMod$ as rewriting in $\RwRules$ modulo $\EqTh$: for all terms $t$ and $t'$, we have $t\rewritesMod t'$ if $l \rewritesto r\in\RwRules$ and there exists a substitution $\sigma$ such that $\sigma(l) =_\B t^*$ where $t^*$ is $t$ or any subterm of $t$ and $t'$ is exactly equal to $t$ with $t^*$ replaced by $\sigma(r)$.

For example, let 
\begin{align*}
\EqTh		&=\emptyset,  \\
\RwRules 	&=\{\mathit{first}(\mathit{tuple}(x, y))\rewritesto x\}, \\
t 			&= \mathit{first}(\mathit{tuple}(\mathit{hash}(x), y)), \\
t'			&=hash(x).
\end{align*}

We have that $t \rewritesMod t'$ using 
\begin{align*}
l 		&= \mathit{first}(\mathit{tuple}(x, y)), \\
r 		&= x, \\
\sigma	&=\{x \gets \mathit{hash}(x), y \gets y\},\\
t^* 	&= t. 
\end{align*}

Note that $\sigma(l) = \mathit{first}(\mathit{tuple}(\mathit{hash}(x), y)) = t = t^*$ and $\sigma(r) = \mathit{hash}(x) = t'$.

Let $\rewritesModCl$ denote the reflexive, transitive closure of $\rewritesMod$. We say that terms $t_1$ and $t_2$ are $\EqTh$\textit{-confluent}, if there exist terms $t_1'$ and $t_2'$ such that $t_1\rewritesModCl t_1'$, $t_2 \rewritesModCl t_2'$ and $t_1' =_\EqTh t_2'$. $\RwRules$ is $\EqTh$-confluent if and only if for any terms $t$, $t_1$ and $t_2$, such that $t\rewritesModCl t_1$ and $t\rewritesModCl t_2$, $t_1$ and $t_2$ are $B$-confluent. Informally, this means that for each term $t$ for which multiple rules apply, we can always apply more rules such that all of the different subresults eventually lead to the same term.

$\RwRules$ is $\EqTh$-\textit{terminating} if and only if there is no infinite rewrite chain for $\rewritesMod$.

$\RwRules$ is $\EqTh$\textit{-coherent} if and only if for any terms $t_1$, $t_2$ and $t_3$ such that $t_1 =_\EqTh t_2$ and $t_1 \rewritesModCl t_3$, there exists a term $t_4$ such that $t_2\rewritesModCl t_4$ and $t_3$ and $t_4$ are $\EqTh$-confluent. In other words, for all terms in the same equivalence class according to $\EqTh$, the order of applications of rules does not matter.

\subsubsection{Unification}
Two terms $t$ and $t'$ are $\EqTh$-unifiable if there exists a substitution $\sigma$ such that $\sigma(t) =_\EqTh \sigma(t')$. In this case, $\sigma$ is referred to as an $\EqTh$-unifier of $t$ and $t'$. A set $\mathcal{S}$ of substitutions is a complete set of $\EqTh$-unifiers for terms $t$ and $t'$ if all substitutions in $\mathcal{S}$ are $\EqTh$-unifiers and for any $\EqTh$-unifier $\sigma'$ of $t$ and $t'$, there exists a substitution $\sigma$ in $\mathcal{S}$, such that $\sigma \leq_\EqTh \sigma'[vars(t) \cup vars(t')]$. More informally, the complete set of $\EqTh$-unifiers of $t$ and $t'$ represents the required $\EqTh$-unifiers of $t$ and $t'$ to build all other $\EqTh$-unifiers for $t$ and $t'$.

An algorithm that generates $\EqTh$-unifiers for arbitrary terms $t$ and $t'$ is said to be complete if it generates a complete set of $\EqTh$-unifiers of $t$ and $t'$. It is additionally said to be finite if the complete set of $E$-unifiers is finite.
For example, take encryption operator ${enc: \texttt{Plaintext} \times \texttt{Key} \rightarrow \texttt{Ciphertext}}$ and ${dec: \texttt{Ciphertext} \times \texttt{Key} \rightarrow \texttt{Plaintext}}$ for sorts \texttt{Msg}, \texttt{Key}, \texttt{Ciphertext} and~\texttt{Plaintext}, where each of those is a subsort of \texttt{Msg}. We have equation $dec(enc(x, k), k) = x$ for $k: \texttt{Key}$ and $x: \texttt{Plaintext}$. Take $t_1 = dec(y, k)$ and $t_2 = z$, with  $y: \texttt{Ciphertext}$, $k: \texttt{Key}$, and $z: \texttt{Plaintext}$. The complete set of unifiers for $t_1$ and $t_2$ is $\{\sigma_1, \sigma_2\}$ with $\sigma_1 = \{y \leftarrow enc(z, k), z \leftarrow z\}$ and $\sigma_2=\{y \gets y, z \gets dec(y, k), k \gets k\}$.

Unification is important in state exploration, since it makes it possible to determine all ways that a certain message could have been created. For example, given a state where the adversary needs to send a specific message to break a security claim, determine how such a message could have been constructed given the transition rules in a rewrite theory $\mathcal{R}$. Unification is much more efficient than enumerating all possible messages an adversary could have sent.

\subsubsection{Decomposition and the finite variant property}
\label{definition:fvp}
Let $\EqTh$ be a $\Sigma$-equational theory. A \textit{decomposition} of $\EqTh$ is a rewrite theory $(\Sig, \B, \Delta)$ such that:
\begin{enumerate}
	\item $\EqTh=\Delta\uplus\B$,
	\item $B$ is regular, sort-preserving and uses top sort variables,
	\item $B$ has a finitary and complete unification algorithm, and
	\item $\Delta$ is sort-decreasing, $B$-confluent, $B$-coherent and $B$-terminating,
\end{enumerate}

Given a $\Sigma$-equational theory $\EqTh$, a decomposition $\RwTh=(\Sig, \B, \Delta)$ of $\EqTh$, and a term $t$, a \emph{variant} of $t$ is a tuple $(t', \sigma)$, such that there exists a term $t''$ and
\begin{enumerate}
	\item	$\sigma(t) \rightarrow_{\B, \Delta}^* t''$,
	\item 	$t''$ is a normal form of $\rightarrow_{\B, \Delta}$, and
	\item   $t'=_\B t''$.	
\end{enumerate}

We define the \textit{variant complexity} as the number of variants for all terms:
\[
	\vc(\RwTh) = \sum_{f\in\FunSet} v(f)
\]
where $\FunSet$ is the set of function symbols in $\Sigma$ and $v(f)$ is the cardinality of the complete set of unifiers \cite{homenc-fvp}.

Then we say that $\RwTh$ has the Finite Variant Property (FVP) if and only if $\vc(\RwTh)$ is finite.
In \cite{escobar-variant-narrowing} it is shown that if a decomposition of an equational theory $\EqTh$ has the finite variant property, there exists a finitary $\EqTh$-unification algorithm, that computes a complete and minimal set of $\EqTh$-unifiers. 

\subsection{Automated Verification}
\label{sssec:automated-verification}
As stated in the introduction, the symbolic model allows us to use automated tools for the formal verification of cryptographic protocols and implementations thereof. Symbolic security analysis tools such as Tamarin \cite{meier2013tamarin}, ProVerif \cite{proverif} and Maude-NPA \cite{maude-npa} allow for fully automating security proofs. Tamarin and Proverif support user-provided equational theories with the FVP that only use two sorts: one to represent general messages and one to represent randomly generated values or user-provided data. In contrast, Maude-NPA allows a user to model a wider range of equational theories. Since the equational theories proposed in this work require extra sorts, we used Maude-NPA for protocol verification.

 \section{Related Work}
\label{sec:related-work}
The first algorithm that could produce a complete set of unifiers modulo an equational theory uses a technique called narrowing \cite{narrowing}. However, this algorithm does not terminate for many equational theories. Jouannaud \cite{split-eq-theory} devised an algorithm that splits up the equational theory into a set of equations $\Delta$ and a set of rewrite rules $B$, as we did in the previous section, and provides a complete set of unifiers under more favourable conditions on $\Delta$ and $B$. This algorithm however, was still rather inefficient due to the fact that many narrowing sequences need to be considered. Escobar et al. \cite{variant-unification} introduced a more efficient procedure that is also complete, called \textit{variant narrowing}, which provides a complete set of unifiers in finite time for any equational theory with the FVP. 

Recall that one of the requirements for a decomposition $(\Sigma, B, \Delta)$ to have the FVP, is for $B$ to have a finitary and complete  unification algorithm. At the moment, there are finitary unification algorithms for:
\begin{enumerate}
	\item commutativity \cite{unificiation-com-assoc}
	\item associativity \cite{unificiation-com-assoc}
	\item commutativity and associativity \cite{unification-com-and-assoc}
	\item idempotence \cite{unification-idemp}
	\item commutativity and associativity and idempotence \cite{unification-com-and-assoc-and-idemp}
	\item abelian group theory \cite{unificiation-abelian}
\end{enumerate}

As shown in the introduction, the problem of unifying modulo the distributivity rule of Boolean logic and the homomorphic property of partially homomorphic encryption are similar. These properties cannot be deconstructed into $B$ and $\Delta$ according to the requirements mentioned in section \ref{definition:fvp}, as shown by Yang et al. \cite{homenc-fvp}. The same authors have proposed approximations of non-FVP theories such that variant narrowing can be applied to unify the combined equational theory of encryption-decryption cancellation and this homomorphic property. One proposal overestimates the homomorphic property by modelling it as an abelian group, which is known to have the FVP. However, if the tool would find an attack using this equational theory, it could be a false attack, because the attack might use mathematical properties of the abelian group that are not applicable to partially homomorphic encryption. Another proposal sets a bound on the number of homomorphic operations that the analysis supports. This is an underapproximation, because an attack that would break the security of a protocol that takes more than the modelled number of homomorphic operations would be missed by a tool that uses this approximation as the equational theory. It therefore underestimates the theory. The equational theory of bounded partially homomorphic encryption of \cite{homenc-fvp} that has the FVP supports a free operator, which means that there are no mathematical properties like associativity, commutativity and identity. We extend this theory for Boolean logic, including its associativity, commutativity and identity properties using a minimal set of operators, namely the \texttt{or}-operator, the \texttt{not} operator and the value $\top$, which represents the value \texttt{True}. With this equational theory, arbitrary computations up to a certain circuit size can be modelled. We additionally prove that this equational theory attains the FVP, which means it can be used in automatic protocol verification by applying variant narrowing for unification. When this equational theory is used within automated protocol verifcation, attacks can be found that perform arbitrary computations, as long as the attack has a circuit size that is smaller than the circuit size chosen in the equational theory.

\section{Modelling Arbitrary Computations in Maude-NPA}
\label{sec:modelling-binary-circuit-in-maude-npa}
In this section, we formalise a decomposition for binary circuits of bounded depth and prove that they model Boolean logic up to the specified depth. Additionally, using an approach similar to the one in \cite{homenc-fvp}, we show how these languages can be implemented as an equational theory with the FVP.

\subsection{Formalising Binary Circuits of Bounded Size}
To model bounded binary circuits, we use the set of sorts $S=\{\bit,\circuit\}$, with $\bit\leq\circuit$. The related function set consists of four functions. 
Firstly, $\top$ is a nullary function that returns a $\bit$, denoting true. Secondly, we have two types of negation: one maps bits to bits and the other maps circuits to circuits. Hence, if $b$ has sort \bit, then $\neg b$ will also have sort $\bit$. Third, we also have disjunction on circuits. Note that using disjunction and negation, we can model all other Boolean operations. We will use infix notation for disjunction in this section. Finally, for variables, we use $\VarSet_{\bit}=\{b_0, b_1, \ldots\}$ and $\VarSet_\circuit=\{c_0,c_1,\ldots\}$. So, as a function set $\FunSet$ we have:
\begin{align*}
	\top~&:~\to\bit \\
	\neg~&:~\bit\to\bit \\
	\neg~&:~\circuit\to\circuit \\
	\lor~&:~\circuit\times\circuit\to\circuit
\end{align*}

For circuits that have been defined in terms of bits, we define the \textit{size} (written $\rank{\cdot}$) of a circuit as the number of disjunctions:
\begin{align*}
	\rank{\top} = \rank{b_i} &= 0, \\
	\rank{\neg{c_i}} &= \rank{c_i}, \\
	\rank{c_i\lor c_j} &= \rank{c_i} + \rank{c_j} + 1.
\end{align*}

We now restrict our attention to terms with a size less than or equal to a number $k$. We denote by $\lang_k$ the language consisting of all circuits of size at most $k$:
\[
\lang_k = \{c\mid \rank{c} \leq k\}.
\]
For each of these languages $\lang_k$, we aim to find an equational theory $\EqTh_k$ with the following properties:
\begin{enumerate}
	\item $\EqTh_k$ captures the behaviour of the circuits in $\lang_k$. That means that two logically equivalent circuits must also be provably equivalent in the equational theory.
	\item $\EqTh_k$ has the FVP. We will prove this by giving a decomposition $(\Sigma, \B_k, \Delta_k)$ and show it satifies each of the properties in section \ref{definition:fvp}
\end{enumerate}
We must first define when two circuits are logically equivalent. We use the standard semantics from propositional logic. Let $\Val: \VarSet_{\bit}\to \{0,1\}$ be a valuation that maps circuits to zero or one. Then we define the relation $\models$ inductively as follows:
\begin{align*}
	\Val\models\top &\text{ always holds}, \\
	\Val\models b_i &\text{ if and only if }\Val(b_i)=1, \\
	\Val\models\neg c_i &\text{ if and only if not }\Val\models{c_i}, \\
	\Val\models c_i\lor c_j &\text{ if and only if }\Val\models{c_i}\text{ or }\Val\models{c_j}.
\end{align*}
Two circuits $c_i$ and $c_j$ are logically equivalent (written as $c_i\iff c_j$) if and only if for all possible valuations $\Val$, we have $\Val\models c_i$ if and only if $\Val\models c_j$. We are then looking for decompositions $(\Sig, \B_k, \Delta_k)$ that are \textit{sound} and \textit{complete} with respect to the circuit in $\lang_k$. That is, for all circuits $c_i,c_j\in\lang_k$: $c_i\iff c_j$ if and only if $c_i$ and $c_j$ rewrite (modulo $\B_k$) to the same normal form.

\subsection{Find the equations for bounded circuits}
\label{sec:findeq}
Intuitively, one could just take the axioms of Boolean logic as rewrite rules. However, this system does not have the FVP. This is easy to see, as the variants of
\[
c_0:\sort{circuit} \vee c_1:\sort{circuit}
\] 
are 
\begin{gather*}
	(c_0'\vee c_1', \{c_0 \leftarrow c_0', c_1 \leftarrow c_1'\}) \\
	(c_0'\vee (c_1'\vee c2'), \{c_0 \leftarrow c_0', c_1 \leftarrow c_1' \vee c_2'\}) \\
	\vdots
\end{gather*}
This produces infinitely many variants. Therefore, we must also restrict our axioms to circuits of bounded size. In the following subsections we give the decompositions for the equational theories $\EqTh_0$ to $\EqTh_3$ which axiomatise $\lang_0$ to $\lang_3$ respectively.

Soundness (equivalence of normal form implies logical equivalence) follows from the fact that all the equations in $\B_k$ and $\Delta_k$ preserve logical equivalence. So for all $l=r$ in our decomposition, we have $l\iff r$. This can easily be checked using truth tables and is left to the reader.

Proving completeness (logical equivalence implies equivalence of normal form) is less trivial. To prove this, we will often employ the following lemma.
\begin{lemma}\label{lem:nf}
	Take a decomposition $(\Sigma, \B, \Delta)$ of an equational theory $\EqTh$. Then $\EqTh$ is complete for a language $\lang$ if
	\begin{enumerate}
		\item all equations in $\B$ and rules in $\Delta$ preserve logical equivalence, and
		\item if two normal forms are not equivalent modulo $\B$ then they are not logically equivalent.
	\end{enumerate}
\end{lemma}
\begin{proof}
	First of all, note that if a decomposition $(\Sigma, \B, \Delta)$ satisfies both conditions, then each circuit has a single normal form modulo $\B$. Namely, suppose a circuit $c$ has two normal forms $c'$ and $c''$ and $c'\neq_\B c''$. Since $\Delta$ preserves logical equivalence we have:
	\[
	c'\iff c\iff c''
	\]
	which violates the second condition for our equation theory. Therefore, each circuit has a single normal form modulo $\B$. Next, we show that the decomposition is complete for $\lang$. Take any two circuits $c_0,c_1\in\lang$ such that $c_0\iff c_1$. Then they both have a single normal form $c'_0$ and $c'_1$. By the first condition, we know that $c'_0\iff c_0$ and $c'_1\iff c_1$. Therefore
	\[
	c'_0\iff c_0\iff c_1\iff c'_1.
	\]
	Therefore, by contrapositive of the second condition, it follows that $c'_0=_\B c'_1$.
\end{proof}

In the following sections, we will provide the equational theories and their decompositions for each language $\lang_0$ to $\lang_3$ respectively. For each decomposition we will prove completeness. As mentioned before, the soundnes proof using truth tables is omitted here for brevity. That the decompositions have all the required properties, including the FVP, will be shown in section \ref{ssec:ThInMaude} using Maude-NPA.

\subsubsection{The language $\lang_0$}
For $\lang_0$ we use a decomposition of $\EqTh_0$ where $\B_0=\emptyset$ and $\Delta_0$ consists of a single rule, namely the double negation elimination (see Table \ref{L0:B}). Using this rule, every circuit in $\lang_0$ can be reduced to either $\pm\top$ or $\pm b_i$ for some $i$. Here, $\pm x$ denotes either $x$ or $\neg x$. Those are the normal forms, numbered 0.1 and 0.2 in Table \ref{L0:nf}.
\begin{table}[!ht]
	\caption{The system $\Delta_0$}
	\centering
	\begin{tabular}{l r c l}
		\hline
		Double Negation & $\neg\neg c_0$ & $\rewritesto$ & $c_0$ \\
		\hline
	\end{tabular}
	\label{L0:B}
\end{table}
\begin{theorem}
	$(\Sig, \B_0,\Delta_0)$ is complete for $\lang_0$.
\end{theorem}
\begin{proof}
	We can show by induction on the number of negations in the terms that every circuit can be rewritten to one of the normal forms of $\lang_0$: $\pm\top$ or $\pm b_i$. Since the rewrite rule preserves logical equivalence and the normal forms are not logically equivalent, the completeness follows from Lemma \ref{lem:nf}.
\end{proof}
\begin{table}[!ht]
	\caption{The normal forms for $\EqTh_0$}
	\centering
	\begin{tabular}{l l}
		\hline
		0.1 & $\pm\top$ \\
		0.2 & $\pm b_0$ \\
		\hline
	\end{tabular}
	\label{L0:nf}
\end{table}
\subsubsection{The language $\lang_1$}
When we allow one disjunction into the language, we let $\Delta_1$ contain the four rules in Table \ref{L1:B}. $\B_1$ contains only commutativity for disjunction.
\begin{table}[!ht]
	\caption{The system $\Delta_1$}
	\centering
	\begin{tabular}{l r c l}
		\hline
		& \multicolumn{3}{c}{All equations from $\Delta_0$} \\
		Annihilation & $c_0\lor\top$ & $\rewritesto$ & $\top$ \\
		Identity & $c_0\lor\neg\top$ & $\rewritesto$ & $c_0$ \\
		Idempotence & $b_0\lor b_0$ & $\rewritesto$ & $b_0$ \\
		Complementation & $b_0\lor\neg b_0$ & $\rewritesto$ & $\top$ \\
		\hline
	\end{tabular}
	\label{L1:B}
\end{table}
\begin{theorem}
	$(\Sig, \B_1,\Delta_1)$ is complete for $\lang_1$.
\end{theorem}
\begin{proof}
	We must show two things. First of all, each equation in $\B_1$ and $\Delta_1$ must preserve logical equivalence. This can easily be checked using truth tables. Secondly, if two normal forms are not equivalent modulo $\B_1$, then they must not be logically equivalent. For this, we claim that all circuits in $\lang_1$ can be reduced to a normal form from Table \ref{L1:nf}. Using truth tables, we see that no two of these normal forms are logically equivalent. Therefore, it remains to show that each circuit can be reduced to a normal form from Table \ref{L1:nf}.
	\begin{table}[!ht]
		\caption{The normal forms for $\EqTh_1$}
		\centering
		\begin{tabular}{l l}
			\hline
			0.1 & $\pm\top$ \\
			0.2 & $\pm b_i$ \\
			1.1 & $\pm(b_0\lor b_1)$ \\
			\hline
		\end{tabular}
		\label{L1:nf}
	\end{table}
	First of all, note that all circuits without disjunction can be rewritten to a normal form in $\lang_0$, since we include $\Delta_0$ in our theory. All other circuits are of the form: ${\pm(c_0\lor c_1)}$ where $c_0,c_1\in\lang_0$. Since $\Delta_0\subset\Delta_1$, we can assume that $c_0$ and $c_1$ are normal forms in $\lang_0$ (otherwise, we could rewrite them to be in normal form). This gives the following case distinction:
	\begin{itemize}
		\item If $c_0=\pm\top$ or $c_1=\pm\top$, then we can reduce the circuit to a normal form in $\lang_0$ using annihilation or identity, together with commutativity.
		\item If $c_0=b_0$ and $c_1=b_1$ for some bits $b_0, b_1$, then we can do a case distinction:
		\begin{itemize}
			\item If $b_0=b_1$, then we can reduce the circuit to $b_0$ (normal form 0.2) using idempotence.
			\item If $b_1=\neg b_0$, then we can use complementation to reduce the circuit to $\top$, normal form 0.1.
			\item Otherwise, the circuit cannot be rewritten using any circuit and we have normal form 1.1.
		\end{itemize}
	\end{itemize}
	Hence the decomposition $(\Sig, \B_1,\Delta_1)$ is complete for $\lang_1$ by lemma \ref{lem:nf}. Note that normal form 1.1 also covers circuits of the form $b_0\lor\neg b_1$, since $\neg b_1$ is also of sort bit.
\end{proof}
Using $\B_1$ and $\Delta_1$, it is always possible to eliminate $\top$ and $\neg\top$ from any non-trivial circuit. Hence, from now on, we will in proofs assume that circuits do not contain $\top$.
\subsubsection{The language $\lang_2$}
For the the decomposition $(\Sig, \B_2,\Delta_2)$ that axiomatises $\lang_2$, we let $\B_2$ contain commutativity and associativity for disjunction and the rules in $\Delta_2$ are given in Table \ref{L2:B}. $\Delta_2$ contains of all rules in $\Delta_1$. In addition, it contains the reduction rules for absorption. We must also include some rule for distributivity. However, the issue is that we cannot add the rule
\[
\neg(b_0\lor b_1)\lor b_2\rewritesto \neg(\neg(\neg b_0\lor b_2)\lor\neg(\neg b_1\lor b_2))
\]
to $\Delta_2$ since the circuit on the right has size 3. However, if $b_2=\pm b_0$ (or $\pm b_1$) then we can reduce the circuit on the right to a circuit of size 2. Therefore, we add the rule DistrCompl to $\Delta_2$. The case where $b_2=\neg b_0$ is covered by the absorption law. 
\begin{table}[!ht]
	\caption{The system $\Delta_2$}
	\centering
	\begin{tabular}{l r c l}
		\hline
		& \multicolumn{3}{c}{All equations from $\Delta_1$} \\
		DistrCompl & $\neg(b_0\lor b_1)\lor b_0$ & $\rewritesto$ & $b_0\lor\neg b_1$ \\
		Absorption & $\neg(b_0\lor b_1)\lor \neg b_0$ & $\rewritesto$ & $\neg b_0$ \\
		Absorption (dual) & $\neg(\neg b_0\lor b_1)\lor b_0$ & $\rewritesto$ & $b_0$ \\
		\hline
	\end{tabular}
	\label{L2:B}
\end{table}
\begin{theorem}
	$(\Sig, \B_2,\Delta_2)$ is complete for $\lang_2$.
\end{theorem}
\begin{proof}
	We deploy the same method as before and show that all circuits in $\lang_2$ can be rewritten to a normal form from Table \ref{L2:nf} and that no two of these normal forms are logically equivalent. The latter can easily be shown using truth tables. 
	\begin{table}[!ht]
		\caption{The normal forms of $\EqTh_2$}
		\centering
		\begin{tabular}{l l}
			\hline
			0.1 & $\pm\top$ \\
			0.2 & $\pm b_0$ \\
			1.1 & $\pm(b_1\lor b_2)$ \\
			2.1 & $\pm(b_1\lor b_2\lor b_3)$ \\
			2.2 & $\pm(\neg(b_1\lor b_2)\lor b_3)$ \\
			\hline
		\end{tabular}
		\label{L2:nf}
	\end{table}
	To prove that all circuits in $\lang_2$ can be rewritten to a normal form, we determine all possible case distinctions we need to cover the circuits in $\lang_2$. All formulas in $\lang_2$ are of the form $c_0\lor c_1$, with $\rank{c_0}=1$ and $\rank{c_1}=0$ (the case where $\rank{c_0}=0$ and $\rank{c_1}=1$ is covered by commutativity). Since $\Delta_0\subset\Delta_1\subset\Delta_2$, we can, without loss of generality, assume that $c_0$ and $c_1$ are normal forms in $\lang_1$ and $\lang_0$ respectively. This leads to the following case distinction:
	\begin{itemize}
		\item Suppose $c_0= b_0\lor b_1$. If $c_1=\pm b_0$ or $\pm b_1$, then we use associativity, together with idempotence or complementation to reduce it to a normal form in $\lang_1$. If $c_1=b_2$, then we obtain normal form 2.1.
		\item Suppose $c_0=\neg(b_0\lor b_1)$. Let us check the possibilities for $c_1$:
		\begin{itemize}
			\item If $c_1=b_0$ or $c_1=b_1$, then we use \mbox{DistrCompl} to reduce the size of the circuit, hence it is equivalent to a normal form in $\lang_0$.
			\item If $c_1=\neg b_0$ or $c_1=\neg b_1$, then we use Absorption to reduce the size of the circuit. Whenever we have a circuit that contains a bit and its negation, we must also add its dual version(s), in which the other occurence has the negation sign.
			\item If $c_1=b_2$, then we have normal form 2.2.
		\end{itemize}
	\end{itemize}
	Since all equations in $\B_2$ and $\Delta_2$ preserve logical equivalence and no two normal forms are logically equivalent, the theorem follows from Lemma \ref{lem:nf}.
\end{proof}
\subsubsection{The language $\lang_3$}
For the decomposition $(\Sig, \B_3,\Delta_3)$ of $\EqTh_3$ for $\lang_3$, we set $\B_3=\B_2$ and $\Delta_3$ contains the rules in Table \ref{L3:B}. We claim this decomposition is complete with respect to $\lang_3$ in the following theorem.
\begin{table}[!ht]
	\caption{The rewrite system $\Delta_3$}
	\centering
	\begin{tabular}{r c l}
		\hline
		\multicolumn{3}{c}{All equations from $\Delta_2$} \\
		\multicolumn{3}{c}{Distributivity (and it's dual):} \\
		$\neg[b_0\lor\neg(b_1\lor b_2)]$ & $\rewritesto$ & $\neg(b_0\lor\neg b_1)\lor\neg(b_0\lor\neg b_2)$ \\
		$\neg[\neg(b_0\lor b_1)\lor\neg(b_0\lor b_2)]$ & $\rewritesto$ & $b_0\lor\neg(\neg b_1\lor \neg b_2)$ \\
		\multicolumn{3}{c}{Idempotence size 1} \\:
		$\neg(b_0\lor b_1)\lor\neg(b_0\lor b_1)$ & $\rewritesto$ & $\neg(b_0\lor b_1)$ \\
		\multicolumn{3}{c}{Distributivity size 1} \\:
		$\neg(b_0\lor b_1)\lor\neg(\neg b_0\lor b_1)$ & $\rewritesto$ & $\neg b_1$ \\
		\multicolumn{3}{c}{Distributivity size 2} \\:
		$\neg(b_0\lor b_1\lor b_2)\lor b_0$ & $\rewritesto$ & $b_0\lor\neg(b_1\lor b_2)$ \\
		\multicolumn{3}{c}{Absorption size 2 (and it's dual)} \\:
		$\neg(b_0\lor b_1\lor b_2)\lor\neg b_0$ & $\rewritesto$ & $ \neg b_0$ \\
		$\neg(\neg b_0\lor b_1\lor b_2)\lor b_0$ & $\rewritesto$ & $b_0$ \\
		\hline
	\end{tabular}
	\label{L3:B}
\end{table}
\begin{theorem}
	$(\Sig, \B_3,\Delta_3)$ is complete with respect to $\lang_3$.
\end{theorem}
\begin{proof}
	The full proof is given in Appendix \ref{app:proofL3}. Here we sketch the most important ideas. The normal forms for $\lang_3$ are given in Table \ref{L3:nf}.
	The strategy is similar to the previous proofs. A circuit of size 3 has the form $c_0\lor c_1$ with either:
	\begin{itemize}
		\item $\rank{c_0}=2$ and $\rank{c_1}=0$, or
		\item $\rank{c_0}=\rank{c_1}=1$.
	\end{itemize}
	However, note that whenever $\vars(c_0)\cap\vars(c_1)=\emptyset$, we obtain normal form 3.1 or 3.5 modulo associativity and commutativity. Also, if $c_0$ is positive, we use associativity and the rules in $\Delta_2$ to reduce the size of the circuit. This leaves us with the case where $c_0$ has an outer negation and $c_1$ has variables in common with $c_0$. This can be handled by a case distinction on the circuits. For the rest, we use a similar case distintion as before. See the appendix \ref{app:proofL3} for the full proof.
	\begin{table}[ht]
		\caption{The normal forms of $\EqTh_3$}
		\centering
		\begin{tabular}{l l}
			\hline
			0.1 & $\pm\top$ \\
			0.2 & $\pm b_0$ \\
			1.1 & $\pm(b_0\lor b_1)$ \\
			2.1 & $\pm(b_0\lor b_1\lor b_2)$ \\
			2.2 & $\neg(b_0\lor b_1)\lor b_2$ \\
			3.1 & $\neg(b_0\lor b_1)\lor\neg(b_0\lor b_2)$ \\
			3.2 & $\pm[\pm(b_0\lor b_1)\lor\pm(b_2\lor b_3)]$ \\
			3.3 & $\pm[\neg(b_0\lor b_1)\lor\neg(\neg b_0\lor\neg b_1)]$ \\
			3.4 & $\pm[\neg(b_0\lor b_1)\lor\neg(\neg b_0\lor b_2)]$ \\
			3.5 & $\pm[\pm(b_0\lor b_1\lor b_2)\lor\pm b_3]$ \\
			3.6 & $\pm[\neg(\neg(b_0\lor b_1)\lor b_2)\lor b_3]$ \\
			\hline
		\end{tabular}
		\label{L3:nf}
	\end{table}
\end{proof}
The normal form $\neg(\neg(b_0\lor b_1)\lor b_2)$ from $\lang_2$ can now be rewritten using distributivity, and hence, must be removed from the list of normal forms.

\subsubsection{Further languages}
For decompositions of higher languages $\lang_k$ with $k>3$, we leave $\B_k=\B_2$ and can construct $\Delta_k$ in a similar way as before.
Here, we give a non-tight bound on the size of these systems, as this gives an indication of the computational cost of running the protocol analyses.

The rules for double negation, annihilation and identity as discussed in the previous subsections of Section \ref{sec:findeq} are applicable to any circuit, and therefore result in $3 = O(1)$ equations.

The rules for idempotence and complementation are only applicable for circuits with odd size.
This results in $O(1)$ equations for circuits of size $i$ with $i$ odd, so $O(k)$ equations in total for a language $\mathcal{L}_k$.
Absorption equations, derived from the equations $\phi \wedge (\phi \vee \psi)$ or $\phi \vee (\phi \wedge \psi)$, where $\phi$ and $\psi$ are disjunctive terms, are applicable to circuits with size at least 2.
For a circuit of size $k$, there are $\lfloor (k-2)/2 \rfloor$ ways the size $|\psi|$ can be chosen as $k-2-|\psi|$ should be even.
Thus, we have $O(i)$ absorption equations for circuits of size $i$, resulting in a total of $O(k^2)$ equations for a language $\mathcal{L}_k$.
Distributivity equations are derived from the different distributivity axioms for Boolean algebra.
The number of ways to choose sizes for the elements in these axioms is $O(4^k)$, therefore we have an exponential bound on the number of rules needed to model all circuits with at most size $k$.

\subsection{Constructing Equational Theories}
\label{ssec:ThInMaude}
In this section, we show how the languages $\mathcal{L}_0$ through $\mathcal{L}_3$ can be converted to equational theories in Maude-NPA \cite{maude-npa} using the rules in $\Delta_0$ through $\Delta_3$ as defined in Tables \ref{L0:B}, \ref{L1:B}, \ref{L2:B} en \ref{L3:B}. The code is listed in appendix \ref{app:maudecode}. In our model, all four equational theories have the same signature $\Sigma$ and the same set of equations $B$, which can be found in Listing \ref{signature_modulo_theory} as Maude code. Three sorts are defined to model the equations, namely \texttt{BitSort}, \texttt{Circuit} and \texttt{Msg}. The latter is required by Maude-NPA as a supersort of all other sorts for protocol verification. The sort \texttt{BitSort} represents bits. When variables of sort \texttt{BitSort} are combined through operators, they always result in the sort \texttt{Circuit}. This ensures that they are $B_2$ terminating, where $B_2$ contains associativity and commutativity. 

In Maude-NPA, the set of equations $B$ is incorporated in the declaration of operators (indicated using \texttt{assoc} and \texttt{comm} respectively), so the associativity and commutativity properties of the \texttt{or}-operator that are part of $B$ are defined in Listing \ref{signature_modulo_theory} as well.

The representations of $\Delta_0$ through $\Delta_3$ in Maude can be found in Listings \ref{l0_maude}, \ref{l1_maude}, \ref{l2_maude} and \ref{l3_maude}. Note that we omitted the rules $\Delta_{i-1}$ for $\Delta_i$ for $i > 0$ for readability, so the complete set of rewrite rules for $\Delta_i$ is the combined rules of listings representing $\Delta_0$ to $\Delta_i$.

For each of the decompositions $(\Sigma, \B_0, \Delta_0)$ for $\EqTh_0$ through $(\Sigma, \B_3, \Delta_3)$ for $\EqTh_3$, the FVP has been confirmed using Maude tooling. We refer the reader back to Section \ref{eq_theory_background} for the defintions of regularity, confluence, coherence and sort-decreasingness. A set of equations $B$ that models associativity and commutativity meets all the criteria for the FVP, since a finitary unification algorithm is known \cite{unification-com-and-assoc} and it is trivial to see that the equations of associativity and commutativity are regular. Using the Church-Rosser checker of the Maude Formal Environment \cite{MFE}, we established the $B$-confluence of $\Delta_0$ through $\Delta_3$ and confirmed that they are sort-decreasing. Using the coherence checker  of the Maude Formal Environment, we determined that they are $B$-coherent. Using the AProVE tool \cite{AProVE} as an external dependency within the termination checker of the Maude Formal Environment, we determined that they are terminating as well. The variant complexity of $\EqTh_0$ through $\EqTh_3$ can be found in Table \ref{tab:verification_time} of Section \ref{sec:results}. Using the definition in Section \ref{definition:fvp}, we can conclude that the equational theories $\EqTh_0$ through $\EqTh_3$ satisfy the Finite Variant Property.

\section{Formal Verification of a Small Circuit} 
\label{sec:verifying-small-circuit}

In this section, we show the results of running experiments using the equational theories for bounded circuits $\EqTh_1$ through $\EqTh_3$. As $\EqTh_0$ does not contain any applications of \texttt{or}, we cannot define a protocol where the equational theory is shown to provide useful insights. We therefore start with $\EqTh_1$, which supports one disjunction. In the protocol we used for the experiments, we simply let a party $A$ load/generate one bit $b$ and send $\neg(b \vee b)$ over the communication line. The property that we verify using Maude-NPA, is that the adversary is able to learn the value of $b$. Since Boolean logic dictates that $\neg(b \vee b)=\neg b$, $\EqTh_1$ and beyond model the equivalence relations for Boolean logic for (at least) one application of the \texttt{or}-operator, and the adversary can apply a \texttt{not}-operator to retrieve $b$, we expect the tool to conclude that running the protocol once results in the adversary learning $b$. This experiment setup is illustrated in Figure \ref{fig:protocol}. The Maude-NPA code is shown in Listing \ref{fig:protocol_maude}. Note that we had to add an operator \texttt{AsBit: Fresh -> BitSort} to model randomly generated bits or data provided by a protocol party that needs to remain confidential in Maude-NPA. The \texttt{Fresh} sort is a special sort in Maude-NPA that represents randomness and cannot have any subsorts or supersorts.
\begin{figure*}[t]
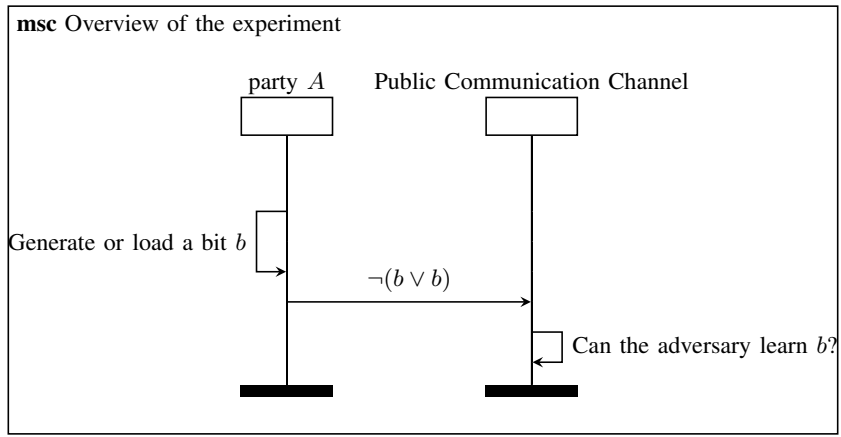

	\centering
	\begin{msc}[small values, environment distance=1cm, instance distance=2cm, first level height=1cm]{Overview of the experiment}
			
			\declinst{client}{party $A$}{}
			\declinst{server}{Public Communication Channel}{}
			
			\mess[side=left]{Generate or load a bit $b$}{client}{client}
			\nextlevel
			\nextlevel
			\nextlevel
			\mess[side=right]{$\neg(b \vee b)$}{client}{server}
			\nextlevel
			\mess[side=right]{Can the adversary learn $b$?   }{server}{server}[1]
		\end{msc}
	\caption{Overview of the protocol that was implemented in Maude-NPA for performance benchmarking.}
	\label{fig:protocol}
\end{figure*}
\begin{lstlisting}[caption={Description of the experiment protocol in Maude-NPA.}, label={fig:protocol_maude}]
vars r1 : Fresh .
	
eq STRANDS-PROTOCOL
= :: r1 ::
[ nil | +(not(or(asBit(r1), asBit(r1)))), nil ]
[nonexec] 
\end{lstlisting}
 In order to accurately model an adversary in Maude-NPA, we need to state rules for the abilities that an adversary has. The adversary should be able to apply \texttt{not}-operators and \texttt{or}-operators to any data, in this model represented by circuits, that are sent over a public communication channel. Lastly, the adversary should be able to generate the value $\top$, which represents the \texttt{True} value in Boolean logic. This allows the adversary to apply arbitrary computations to any data that they obtain. Listing \ref{adversary_rules} shows how these rules are encoded into Maude-NPA.

\begin{lstlisting}[caption={Adversary capabilities as modelled in Maude-NPA.}, label={adversary_rules}]
	vars C0 C1 : Circuit .
	 
	eq STRANDS-DOLEVYAO
	= :: nil ::[ nil | -(C0), -(C1), +(or(C0, C1)), nil ] &
	:: nil ::[ nil | -(C0), +(not(C0)), nil ] &
	:: nil ::[ nil | +(top), nil]
	[nonexec] .
\end{lstlisting}

We want to identify how the performance of the verification scales with increasing sizes for the equational theory. Therefore we analyse this simple protocol for the different equational theories $\EqTh_1$, $\EqTh_2$, and $\EqTh_3$.

\section{Results}
\label{sec:results}

In this section, we list the results of our experiments. We verified that Maude-NPA produces the expected output and list the time it takes for Maude-NPA to produce the expected output for the protocol in Figure \ref{fig:protocol} using the decompositions of the equational theories $\EqTh_1 = (\Sigma, \B_2, \Delta_1)$, $\EqTh_2=(\Sigma, \B_2, \Delta_2)$ and $\EqTh_3=(\Sigma, \B_2, \Delta_3)$ using a laptop with an Intel Core i7-8665U CPU and 16GB RAM. The verification times are the mean for $50$ separate runs of the respective scripts. Additionally, we show the variant complexity of $\EqTh_0$ through $\EqTh_3$. These results can be found in Table \ref{tab:verification_time}. 
\begin{table}[htbp]
	\caption{Time to verify the protocol in Figure \ref{fig:protocol} and variant complexity for $E_0$ through $E_3$ .}
	\label{tab:verification_time}
	\begin{center}
	\begin{tabular}{|r|r|r|}
		\hline
		Equational Theory	&	Verification Time	& Variant Complexity \\ \hline
		$\EqTh_0$			& -						& $4$ \\ \hline
		$\EqTh_1$           & $1.233s$ 				& $11$ \\ \hline
		$\EqTh_2$           & $1.823s$  			& $19$\\ \hline
		$\EqTh_3$           & $21.348$  			& $30$ \\ \hline
	\end{tabular}
	\end{center}
\end{table}
\section{Conclusion}
\label{sec:conclusion}
In this paper, we have presented the first class of equational theories for arbitrary computations up to a certain depth. We prove that this class of equational theories attains the FVP, which makes it possible to use these equational theories in combination with automatic protocol verification in the symbolic model to reason about adversaries that perform arbitrary computations on data sent over a communication channel. We have proven that this class of equational theories up to size $3$ is equivalent to Boolean logic for circuits up to size $3$, ensuring that this class of equational theories correctly models Boolean logic and therefore arbitrary computations. We additionally provide a benchmark for the verification time to verify a simple logical statement within Maude-NPA with these equational theories and show the variant complexity of these equational theories. 

We conjecture that equational theories $\EqTh_0$, $\EqTh_1$, $\EqTh_2$ and $\EqTh_3$ can be extended to support Boolean circuits of arbitrary size. That would technically mean that any algorithm, including cryptographic primites like encryption algorithms and digital signature algorithms, can be modelled in automated protocol verification by rewriting it as a binary circuit. However, the circuit size of such cryptographic algorithms is in the order of millions, which would make it less feasible considering that the verification time scales poorly with the supported circuit size in our approach. 

It is additionally possible to extend $\EqTh_0$ through $E_3$ to equational theories that model fully-homomorphic encryption and multi-party computation techniques like secret sharing, while still maintaining the FVP. The scalability of our approach would then be less of a problem, since the equational theory would only need to support the circuit size of the function to be homomorphically evaluated, instead of the circuit size of the (cryptographic) algorithms. Additionally, the size needs to be extended to allow the tool to find attacks, but that is inherent to this approach. The function to be evaluated could be orders of magnitude smaller than cryptographic algorithms in terms of circuit size, which could make our approach suitable for these use cases. 

Moreover, in this work, we focused on equational theories with the FVP, such that automatic protocol verification tools can use a generic unification algorithm, specifically variant unification \cite{variant-unification}. In previous work, a protocol-specific unification algorithm for homomorphic encryption was used to drastically improve the complexity of verifying protocols using this equational theory \cite{unification_hom_enc} and some unification algorithms for distributivity are known \cite{boolean_unification}. In the future, it would be interesting to see if the scalability in terms of circuit size for performance can be improved by implementing a specific unification theory for Boolean logic that is suitable for protocol verification.

\section*{Acknowledgements}
We thank our colleagues Vincent Dunning and Maaike van Leuken for their support in running some of the tooling and for their reviews.

\bibliographystyle{plain}
\bibliography{references}
\appendices
\section{Data Availability}
\label{app:data}
The research has been conducted using various scripts written in \texttt{Maude}. These scripts and instructions on how to use them to reproduce our results can be found at \url{https://zenodo.org/records/16419234}.
\section{Proof of completeness $\B_3$}
\label{app:proofL3}
Here we give a detailed proof for the completeness of $\B_3$ with respect to $\lang_3$. 

Using the same strategy as before, we must show that each formula is logically equivalent to a normal form. As mentioned in the text, we do this by a case distinction. We can ignore the cases where both disjuncts have no variables in common (since this will always result in a new normal form).

Consider any circuit $c_0\lor c_1$ of rank 3. Then either $\rank{c_0}=1$ and $\rank{c_1}=1$ or we have $\rank{c_0}=2$ and $\rank{c_1}=0$.

Let us start with the case where $\rank{c_0}=\rank{c_1}=1$. In the case where one of the disjuncts is positive (has no outer negation), we can use the rules from $\Delta_2$ to reduce the rank of the formula. Therefore, we only need to cover the cases where both $c_0$ and $c_1$ are of the form $\neg(b_i\lor b_j)$ for some indices $i,j$. Let $c_0=\neg(b_0\lor b_1)$. We now perform a case distinction on $c_1$.
	\begin{itemize}
		\item Suppose $c_1=\neg(b_0\lor b_1)$. Then we use idempotence of rank 1 to reduce the circuit to normal form 1.1.
		\item Suppose $c_1=\neg(\neg b_0\lor b_1)$. Then we can reduce the circuit to normal form 0.2 by Distributivity 1. The case where $c_1=\neg(b_0\lor\neg b_1)$ is equivalent to this case by commutativity  and symmetry. Also the dual version of Distributivity 1 is equivalent due to symmetry.
		\item Suppose $c_1=\neg(\neg b_0\lor\neg b_1)$. Then we obtain normal form 3.3. 
		\item Suppose $c_1=\neg(b_0\lor\pm b_2)$. Then we obtain to normal form 3.1. The same holds for $c_1=\neg(b_1\lor\pm b_2)$.
		\item Suppose $c_1=\neg(\neg b_0\lor\pm b_2)$. Then we obtain normal form 3.4. The same holds for $c_1=\neg(\neg b_1\lor\pm b_2)$.
	\end{itemize}

Let us now consider the case where $\rank{c_0}=2$ and $\rank{c_1}=0$. Again, we only consider cases where $c_0$ and $c_1$ are normal form. Therefore, formula $c_1$ is always of the form $\pm b_i$ for some $i$ (we ignore the case where $c_1=\pm\top$ since we include $\B_1$ in the set of equations and can always eliminate $\pm\top$). 
\begin{itemize}
	\item Suppose $c_0=\neg(b_0\lor b_1\lor b_2)$. If $c_1=b_i$ with $i\in\{0,1,2\}$, then we use Distributivity 2 to reduce it to normal form 2.2. If $c_1=\neg b_i$, with $i\in\{0,1,2\}$, then we use Absorption 2 to reduce the formula to normal form 0.2.
	\item Suppose $c_0=\neg(\neg(b_0\lor b_1)\lor b_2)$. We do a case distinction on $c_1$:
	\begin{itemize}
		\item If $c_1$ is $\pm b_0$ or $\pm b_1$, reduce it to a normal form using the following derivations. First we show the derivation for $b_0$ (the case for $b_1$ is identical):
		\begin{align*}
			&\neg(\neg(b_0\lor b_1)\lor b_2)\lor b_0 \\
			&\rewritesMod \neg(\neg b_0\lor b_2)\lor\neg(\neg b_1\lor b_2) \lor b_0 \tag{Distr} \\
			&\rewritesMod b_0\lor\neg(\neg b_1\lor b_2) \tag{Absorption (dual)}
		\end{align*}
		Next, we show the derivation for $\neg b_0$ (case $\neg b_1$ is identical):
		\begin{align*}
			&\neg(\neg(b_0\lor b_1)\lor b_2)\lor\neg b_0 \\
			&\rewritesMod \neg(\neg b_0\lor b_2)\lor\neg(\neg b_1\lor b_2)\lor\neg b_0 \tag{Distributivity} \\
			&\rewritesMod \neg b_0\lor\neg b_2\lor\neg(\neg b_1\lor b_2) \tag{DistrCompl} \\
			&\rewritesMod \neg b_0\lor\neg b_2 \tag{Absorption}
		\end{align*}
		\item If $c_1=b_2$, then we use the following derivation to reduce it to a normal form.
		\begin{align*}
			&\neg(\neg(b_0\lor b_1)\lor b_2)\lor b_2 \\
			&\rewritesMod \neg(\neg b_0\lor b_2)\lor\neg(\neg b_1\lor b_2) \lor b_2 \tag{Distributivity} \\
			&\rewritesMod \neg(\neg b_0\lor b_2)\lor b_1\lor b_2 \tag{DistrCompl} \\
			&\rewritesMod b_0\lor b_1\lor b_2 \tag{DistrCompl}
		\end{align*}
		A similar derivation holds for $\neg b_2$ using Absorption 2:
		\begin{align*}
			&\neg(\neg(b_0\lor b_1)\lor b_2)\lor\neg b_2 \\
			&\rewritesMod \neg(\neg b_0\lor b_2)\lor\neg(\neg b_1\lor b_2) \lor\neg b_2 \tag{Distributivity} \\
			&\rewritesMod \neg(\neg b_0\lor b_2)\lor\neg b_2 \tag{Absorption} \\
			&\rewritesMod \neg b_2 \tag{Absorption}
		\end{align*}
	\end{itemize} 
\end{itemize}
Hence, we conclude that the all circuits in $\lang_3$ can be rewritten using $\RwRules_3$ to normal forms such that not two normal forms are logically equivalent to each other. Therefore the equational theory $(\Sig,\B,\RwRules_3)$ is complete with respect to $\lang_3$.

\section{Maude code}
\label{app:maudecode}

\begin{lstlisting}[caption={Signature $\Sigma$ and equations $B$ in Maude.}, label={signature_modulo_theory}]
	sorts BitSort Circuit Msg .
	subsort BitSort < Circuit .
	subsort Circuit < Msg .
	op not : BitSort -> BitSort .
	op not : Circuit -> Circuit .
	op or : Circuit Circuit -> Circuit [assoc comm] .
	op top : -> BitSort .
\end{lstlisting}

\begin{lstlisting}[caption={Rewrite rules $\Delta_0$.}, label={l0_maude}]
	vars B0 B1 B2 B3 : BitSort .
	vars C0 C1 C2 : Circuit .

	*** Double negation
	eq not(not(c_0)) = c_0  .

\end{lstlisting}

\begin{lstlisting}[caption={Rewrite rules for $\Delta_1$ without the rules for $\Delta_0$.}, label={l1_maude}]
	vars B0 B1 B2 B3 : BitSort .
	vars C0 C1 C2 : Circuit .

	*** Annihilation
	eq or(top, C0) = top  .
	*** Identity
	eq or(not(top), C0) = C0  .
	*** Idempotence size 0
	eq or(B0, B0) = B0  .
	*** Complementation
	eq or(B0, not(B0)) = top  .

\end{lstlisting}

\begin{lstlisting}[caption={Rewrite rules $\Delta_2$, without the rules for $\Delta_1$.}, label={l2_maude}]
	vars B0 B1 B2 B3 : BitSort .
	vars C0 C1 C2 : Circuit .

	*** DistrCompl
	eq or(not(or(B1, B0)), B0) = or(B0, not(B1))  .
	*** Absorption
	eq or(not(or(B1, B0)), not(B0)) = not(B0)  .
	*** Absorption dual
	eq or(not(or(not(B0), B1)), B0) = B0  .
\end{lstlisting}

\begin{lstlisting}[caption={Rewrite rules $\Delta_3$ without the rules for $\Delta_2$.}, label={l3_maude}]
	vars B0 B1 B2 B3 : BitSort .
	vars C0 C1 C2 : Circuit .

	*** Distributivity
	eq not(or(B0, not(or(B1, B2)))) = or(not(or(B0, not(B1))), not(or(B0, not(B2))))  .
	*** Distributivity dual
	eq not(or(not(or(B0, B1)), not(or(B0, B2)))) = or(B0, not(or(not(B1), not(B2))))  .
	*** Idempotence size 1
	eq or(not(or(B0, B1)), not(or(B0, B1))) = not(or(B0, B1))  .
	*** Distributivity size 1
	eq or(not(or(B0, B1)), not(or(not(B0), B1))) = not(B1)  .
	*** Distributivity size 2
	eq or(not(or(B0, B1, B2)), B0) = or(B0, not(or(B1, B2)))  .
	*** Absorption size 2
	eq or(not(or(B0, B1, B2)), not(B0)) = not(B0)  .
	*** Absorption size 2 dual
	eq or(not(or(not(B0), B1, B2)), B0) = B0  .
\end{lstlisting}

\end{document}